%% file: main.tex
\documentclass[12pt,letterpaper]{article}

\usepackage{amssymb,amsmath,amssymb,bbm,mathtools} 

\usepackage[margin=1in]{geometry}
\usepackage{times}

\usepackage[ruled]{algorithm2e} % For algorithms

\usepackage{amsthm}
\usepackage{thmtools,thm-restate}

\usepackage[numbers]{natbib}

\setcitestyle{acmnumeric}
 
\PassOptionsToPackage{hyphens}{url}

\usepackage{eqparbox,array}
\usepackage{nicefrac}
\usepackage{xspace}

\newtheorem{theorem}{Theorem}[section]
\newtheorem{lemma}[theorem]{Lemma}

\newtheorem{corollary}[theorem]{Corollary}

\newtheorem{example}[theorem]{Example}

\newtheorem{observation}[theorem]{Observation}

\newtheorem*{theorem*}{Theorem}

\newtheorem{definition}[theorem]{Definition}

\usepackage[capitalise,noabbrev]{cleveref}

\makeatletter
\renewcommand{\paragraph}{%
  \@startsection{paragraph}{4}%
  {\z@}{1.0ex \@plus 1ex \@minus .2ex}{-1em}%
  {\normalfont\normalsize\bfseries}%
}
\let\oldnl\nl% Store \nl in \oldnl
\newcommand{\nonl}{\renewcommand{\nl}{\let\nl\oldnl}}% Remove line number for one line
\makeatother

\usepackage{enumerate}
\usepackage{enumitem}
\usepackage[T1]{fontenc}
\usepackage{bbm}
\usepackage{comment}

\renewcommand{\hat}{\widehat}

\DeclareMathOperator*{\argmin}{argmin}

\newcommand{\NSW}{\operatorname{NSW}}
\newcommand{\W}{\operatorname{W}}
\theoremstyle{remark}

\newcommand{\Real}{\mathbb{R}}

\newcommand{\sw}{\operatorname{SW}}

\Crefname{ALC@unique}{Line}{Lines}

\usepackage[breakable, theorems, skins]{tcolorbox}
\DeclareRobustCommand{\mybox}[2][gray!20]{%
	\begin{tcolorbox}[   %% Adjust the following parameters at will.
		breakable,
		left=0pt,
		right=0pt,
		top=0pt,
		bottom=0pt,
		colback=#1,
		colframe=#1,
		enlarge left by=0mm,
		boxsep=7pt,
		arc=0pt,outer arc=0pt,
		]
		#2
	\end{tcolorbox}
}

\title{\bfseries Two Birds With One Stone:\\ Fairness and Welfare via Transfers}

\author{
	Vishnu V. Narayan\thanks{McGill University. {\tt vishnu.narayan@mail.mcgill.ca}}
	\and
	Mashbat Suzuki\thanks{McGill University. {\tt mashbat.suzuki@mail.mcgill.ca}}
	\and
	Adrian Vetta\thanks{McGill University. {\tt adrian.vetta@mcgill.ca}}
}
\date{}

\begin{document}	

\maketitle

\begin{abstract}
We study the question of dividing a collection of indivisible goods amongst a set of agents. The main objective of 
research in the area is to achieve one of two goals: fairness or efficiency. On the fairness side, \textit{envy-freeness} is 
the central fairness criterion in economics, but envy-free allocations typically do not exist when the goods are indivisible. 
A recent line of research shows that envy-freeness \textit{can} be achieved if a small quantity of a homogeneous divisible 
good (money) is introduced into the system, or equivalently, if transfer payments are allowed between the agents. A natural 
question to explore, then, is whether transfer payments can be used to provide high \textit{welfare} in addition to envy-freeness, and if 
so, how much money is needed to be transferred.

We show that for general monotone valuations, there always exists an allocation with transfers that is envy-free and whose 
Nash social welfare (NSW) is at least an $e^{-1/e}$-fraction of the optimal Nash social welfare. Additionally, when the agents have 
additive valuations, an envy-free allocation with negligible transfers and whose NSW is within a constant factor 
of optimal can be found in polynomial time. Consequently,  we demonstrate that the seemingly incompatible objectives of fairness and high welfare 
can be achieved simultaneously via transfer payments, even for general valuations, when the welfare objective is NSW. On the other hand, we show that a similar result is impossible for utilitarian 
social welfare: \textit{any} envy-freeable allocation that achieves a constant fraction of the optimal welfare requires non-negligible transfers. 
To complement this result we present algorithms that compute an envy-free allocation with a given target welfare and with bounded transfers.
\end{abstract}

\input{intro}

\input{prelim}

\input{FairDivisionTransfer}

\input{NSW}
\input{TPSW}

%\input{conclusion}

%\newpage
\bibliographystyle{ACM-Reference-Format}
\bibliography{main.bib}

\end{document}

%% file: intro.tex
\section{Introduction}\label{sec: intro}

The question of how to divide a collection of items amongst a group of agents has remained of central importance 
to society since antiquity. Real-world examples of this problem abound, ranging from the division of land and inherited 
estates, border settlements, and partnership dissolutions, to more modern considerations such as the division of the 
electromagnetic spectrum, distribution of computational resources, and management of airport traffic. The predominant objective 
of research in this area is to study the existence of allocations that achieve one of two broad goals: \textit{fairness} 
or \textit{efficiency}. At a high level, the fairness goal is to ensure that each agent receives its due share of the items, 
and the efficiency goal is to distribute the items in a way that maximizes the aggregate utility achieved by all of the agents.

The study of fair division burgeoned in the decades following its formal introduction by Banach, Knaster and Steinhaus~\cite{Ste48}, and 
most of the early literature focused on the \textit{divisible} setting, where a single heterogeneous divisible item (conventionally, a cake) is 
to be fairly shared among a set of agents with varying preferences over its pieces. The second half of the last century saw the creation of 
precise mathematical definitions for various fairness notions, and \textit{envy-freeness}, where every agent prefers its piece to any piece 
received by another agent, has since emerged as the dominant fairness criterion in economics. Early non-constructive results proved that, 
under mild assumptions, envy-free allocations always exist in the divisible setting~(\citet{Str80,Su99}), and ensuing work produced finite 
and bounded protocols for computing these allocations~(\citet{BT95,AM16}).

More recently, research has focused on the \textit{indivisible} setting, where each item in a collection must be allocated as a whole to 
some agent. It appears at first that envy-freeness cannot be achieved in this setting; consider the simple example of two agents and 
one item, where one agent is left envying the other in any allocation. Consequently, a common theme in the indivisible setting is the 
study of weaker fairness guarantees such as \textit{EF1} and \textit{approximate-MMS}~\cite{LMM04,Bud11,KPW18}.

But is it necessary to restrict ourselves to these weaker guarantees? A recent line of research shows, rather surprisingly, that it is 
possible to achieve canonical envy-freeness even in the indivisible setting simply by adding to the system a small quantity of 
a \textit{divisible good}, akin to money~\cite{HS19,BDN20}, or equivalently by allowing the agents to make transfer payments between themselves. These transfer payments can always be made alongside an allocation of the indivisible items such that the result is envy-free. In this work, we ask and answer a natural follow-up question: 
can this tool be made to do more? Can we use it to simultaneously guarantee full envy-freeness while also 
achieving high \textit{welfare}, and if so, how much in total transfer payments do we need for this?

%But is it necessary to restrict ourselves to these weaker guarantees? A recent line of research shows, rather surprisingly, that it is 
%possible to achieve canonical envy-freeness even in the indivisible setting simply by adding to the system a small quantity of 
%a \textit{divisible good}, akin to money~\cite{HS19,BDN20}. This \textit{subsidy} can always be divided amongst the agents 
%alongside the indivisible items in an allocation that is envy-free. In this work, we ask and answer a natural follow-up question: 
%can this tool can be made to do more? Can we use it to simultaneously guarantee canonical envy-freeness while also 
%achieving high \textit{efficiency}, and if so, how much subsidy is required?

\subsection{Related Work}

The formal origin of fair division dates back to the 1940s, when Banach, Knaster and Steinhaus~\cite{Ste48} 
devised the Last Diminisher procedure to fairly divide a cake among $n$ agents. Their fairness objective was \textit{proportionality}, in 
which each agent receives a piece of value at least $\frac{1}{n}$ of the value of the entire cake to that agent. The pursuit of proportional 
cake divisions in different settings led to the creation of popular algorithmic paradigms for cake-cutting such as the moving-knives 
procedures~\cite{DS61,Str80,RW98}. In the following decades, \textit{envy-freeness}~(\citet{GS58,Fol67}) emerged as the canonical 
fairness solution. When the valuation functions are additive, envy-freeness implies proportionality and is therefore a stronger fairness 
property. A chain of subsequent results culminated in the discovery of finite-time~\cite{BT95} and bounded-time~\cite{AM16} protocols for 
finding envy-free allocations in the divisible setting.

The research efforts of the fair division community have undergone two major shifts in recent years. The first of these is an increased 
focus on \textit{economic efficiency}. The most common type of economic efficiency is \textit{Pareto efficiency}, in which no agent's 
allocation can be improved without making some other agent worse off. A classical result of \citet{Var74} shows that in the divisible 
setting there always exists an allocation that is both envy-free and Pareto efficient. In fact, such an allocation can be computed in 
polynomial time~\cite{DPS08}. A different notion of efficiency arises when we maximize a \textit{welfare function} that measures the 
aggregate utility of all agents. The most common welfare functions studied in the associated literature are the \textit{utilitarian} social 
welfare (or simply the social welfare), which measures the sum of the agents' valuations, and the \textit{Nash} social welfare, which 
measures the geometric mean of these valuations. In the divisible setting, \citet{BCH12} and \citet{CLP11} study the computational 
problem of maximizing the social welfare under proportionality and envy-freeness constraints respectively.

The second shift is towards the study of the indivisible setting, where $m$ items are to be integrally divided amongst $n$ agents. 
Since neither envy-freeness nor proportionality can now be guaranteed, a natural alternative is to provide relaxations or 
approximations of them. One such relaxation is the \textit{EF$k$} guarantee. An allocation is \textit{envy-free up to $k$ goods}, 
or EF$k$, if no agent envies another agent's bundle provided some $k$ goods are removed from that bundle. The EF1 guarantee is 
particularly notable, as EF1 allocations exist and can be computed in polynomial time if the valuation functions are monotone~\cite{LMM04}. 
Two similar relaxations exist for proportionality, namely the \textit{Prop1} guarantee and the \textit{maximin share} guarantee, the latter of 
which is a natural extension of the two-agent cut-and-choose protocol~\cite{Bud11}. A large body of research produced over the last 
decade aims to achieve these guarantees or approximations thereof~\cite{KPW18,GHS18,BCF19,GT20}, including many results that 
show that these fairness guarantees can be achieved alongside Pareto efficiency~\cite{BKV18} or high Nash social 
welfare~\cite{BKV18, CKM19, GM19}.

The problem of achieving high utilitarian social welfare under fairness constraints was formally introduced by \citet{CKK09}. The \textit{price of 
fairness} (that is, of envy-freeness, EF1, or any other fairness criterion) of an instance is defined as the ratio of the social welfare of an 
optimal allocation without fairness constraints, to the social welfare of the best fair allocation. Intuitively, it measures the necessary 
worst-case loss in efficiency when we add fairness constraints. \citet{CKK09} present bounds on the price of fairness (proportionality, 
envy-freeness and equitability) in both the divisible and indivisible settings; we remark, however, that their results for the indivisible 
case only consider the special set of instances for which the associated fair allocations exist. For the divisible setting, \citet{BFT11} 
showed that the bounds of \cite{CKK09} are tight. Followup work on the price of fairness in the indivisible setting by \citet{BLM19} 
and \citet{BBS20} considers only the relaxed fairness guarantees (such as EF1 and $\frac{1}{2}$-MMS) that are always achievable in the indivisible setting.

In now classical work, \citet{Sve83}, \citet{Mas87}, and \citet{TT93} studied the indivisible item setting and asked if it is always 
possible to achieve an envy-free allocation simply by introducing a small quantity of a divisible good, akin to money, alongside the 
indivisible items. Their positive results were mirrored in followup work by \citet{ADG91,Ara95,Kli00} and \citet{HRS02} which showed 
for various settings the existence of an envy-free \textit{allocation with subsidy}. However, all of the above papers considered the 
restricted case where the number of items, $m$, is at most the number of agents $n$. It was only recently that \citet{HS19} extended 
these results to the general $m$-item setting, showing that an envy-free allocation with subsidy always exists in general. \citet{BDN20} 
followed this up with upper bounds on the amount of money sufficient to support an envy-free allocation in all instances. Surprisingly, 
when the valuation functions are scaled so that the marginal value of an item is at most one dollar to any agent, at most $n-1$ dollars 
in the additive case and at most $O(n^2)$ dollars in the general monotone case are always sufficient to eliminate envy~\cite{BDN20}.
Note that the maximum required subsidy is {\em independent} of the number $m$ of items, an observation of particular relevance to our work. Several recent papers study the problem of achieving envy-freeness alongside other properties via subsidies and transfers, including \citet{Azi20,GIK21}.

\subsection{Results and Contributions}

A salient question is whether the two ideas exposited in the prior discussion can be combined: is it 
possible to find an allocation with subsidy that is simultaneously envy-free \textit{and} guarantees high welfare? If so, how 
much subsidy is sufficient to achieve this? These questions are the focus of this paper.

%In this paper, our first goal is to provide bounds on the total welfare of an envy-free allocation 
%with transfer payments. Our second goal is to study the amount of transfer payments sufficient to achieve a given target welfare in any instance. 
%Our measure of welfare is the very general class of $\rho$-mean welfare functions, with particular focus on the two most important
%special cases, namely the Nash social welfare and utilitarian social welfare functions.

Thus, one contribution of our work is to extend the literature on subsidies and their application.
However, rather than subsidies, we analyze the related concept of {\em transfer payments} between the agents for two reasons.
First, a subsidy is an external source of added utility which, in the context of welfare, would bias any subsequent 
comparisons with the welfare-maximizing allocation \textit{without} subsidies. A transfer payment is neutral in this regard.
Second, subsidies require an external agent willing to fund the mechanism -- a typically unrealistic hope. In contrast, transfer 
payments require the consent only of the agents who are already willing participants in the mechanism. Provided the cost of the payments are outweighed by the benefits of participation then giving consent is reasonable.
We remark that subsidies and transfers are in a sense interchangeable. Given an envy-free allocation with subsidies, subtracting the 
average subsidy from each agent's individual payment gives payments which sum to zero, that is, transfer payments. Conversely,
given transfer payments, adding an appropriate fixed amount to each payment induces non-negative subsidy 
payments.\footnote{Of course, whilst the correspondence between subsidies and transfers is simple, the switch to transfer payments does have 
a technical drawback: because transfer payments do not provide an (unnatural) external boost to welfare, obtaining welfare guarantees for the 
case of transfers is generally harder than for the case of subsidies.}

A second contribution is to extend the research on the price of fairness. 
Specifically, we impose no {\em balancing constraint} on the valuation functions of the agents. To understand this,
note that a common assumption in the price of fairness literature is that the
valuation function of each agent is scaled so that the value of the grand bundle of items is \textit{equal} for all agents. 
In the context of fairness, this scaling is benign because it has no affect on the most widely used measures of fairness. For example, it 
does not change the (relative) envy between any pair of agents. Indeed, this assumption is inspired by the literature on envy-free cake-cutting, 
where the value of the entire cake is usually assumed to be $1$ for every agent. 
However, in the context of efficiency or welfare, this scaling is {\em not} benign.
Such a scaling can dramatically alter the welfare of any allocation. In particular, this scaling restricts attention only to balanced
instances, where agents are of essentially equal importance in generating welfare. This is important because it is the elimination of unbalanced
instances that allows non-trivial bounds on the price of fairness to be obtainable~(\cite{BLM19,BBS20}).
Indeed, as will be seen in this paper, it is the unbalanced instances that are typically the most problematic in obtaining both fairness and
high welfare. So, as stated, we study case of general valuation functions with no balancing constraints.

We do, however, make the standard assumption in the literature on subsidies (\cite{HS19, BDN20}), and assume that the maximum marginal value for an item for any 
agent is always at most one dollar. We emphasize that this assumption induces no loss of generality with regards to the valuation functions. 
In particular, this assumption is benign in respect to both fairness and welfare. 
It does not affect the relative envy between agents, and it does not affect the welfare of an allocation (as all valuations can be scaled down uniformly). We make this assumption for simplicity, since expressing the transfers in dollar amounts allows for a consistent comparison with earlier work on the topic. Equivalent bounds for the original instance can be recovered by multiplying these expressions by the maximum marginal value of an item for any agent.

%\textbf{[Added the preceding paragraph. I suppose it's best to also discuss here what we mean when we use negligible and non-negligible throughout; it follows the same ``theme'' or general assumption as the last paragraph: n << m. Could you add this part Adrian?]}

We now present the main results that appear in the paper. We study the trade-off between fairness and efficiency in the presence of 
transfer payments for the class of $\rho$-mean welfare functions, with particular focus on the two most important special cases, namely the 
Nash social welfare and utilitarian social welfare functions. An allocation is \textit{envy-freeable} if it can be made envy-free with the addition of subsidies (or, equivalently, transfer payments). Our first 
observation is that to achieve both fairness and high welfare, it is not sufficient to simply find an envy-freeable allocation -- making transfer 
payments is necessary. In fact, no non-zero welfare guarantee is achievable for all $\rho$ without considering transfers in the computation of the welfare.
Letting $\W^{\rho}$ denote $\rho$-mean welfare, we have:
\begin{observation}\label{obs: transfers}
For any $\epsilon>0$, there exist instances where the welfare of every envy-freeable allocation $A$ satisfies
	$$
	\frac{\W^{\rho}(A)}{\W^{\rho}(A^*)} \le \epsilon 
	$$ 
\end{observation}
Here $A^*$ is the welfare-maximizing allocation. The observation applies even in the case of additive valuations
with Nash social welfare functions. Consequently, the focus on allocations with transfers, and on bounding total transfers, is justified.
For $\rho$-mean welfare functions, we show that positive welfare guarantees are achievable with transfers.
\theoremstyle{definition}
\newtheorem*{thm:EMatch}{Corollary \ref{cor: EnvyMatch}}
\begin{thm:EMatch}
For subadditive valuations, there exists an envy-free allocation with transfers $(A,t)$ such that 
$$
\frac{\W^{\rho}(A,t)}{\W^{\rho}(A^*)} \ge \frac{1}{n} 
$$ 
and with a total transfer $\sum_{i}|t_i|$ of at most $2n^2$. This allocation can be computed in polynomial time.
\end{thm:EMatch}
Here $n$ is the number of agents. Note that the total transfer is independent of the number $m$ of items.
This implies, as $m$ grows, that the transfer payments are negligible in terms of the number of items
(and of total welfare). In particular, our ultimate objective is to obtain both envy-freeness
and high welfare using negligible transfers.

Of course, the welfare guarantee of $\frac{1}{n}$ does not signify high welfare. So we investigate whether
improved bounds can be obtained for the important special cases of $\rho=0$ (Nash social welfare) and $\rho=1$
(utilitarian social welfare).
Strong guarantees on welfare can be obtained for Nash social welfare. Specifically, there exists an envy-free allocation with transfers with a Nash social welfare that is at least an $e^{-1/e}\approx 0.6922$ fraction of the optimal welfare. 
\theoremstyle{definition}
\newtheorem*{thm:N1}{Theorem \ref{thm: NSW}}
\begin{thm:N1}
	For general valuations, there exists an envy-free allocation with transfers $(A,t)$ such that 
	$$
	\frac{\NSW(A,t)}{\NSW(A^*)}\geq e^{-1/e} 
	$$
\end{thm:N1}
Furthermore, for additive valuations, such constant factor welfare guarantees can be obtained with negligible transfer payments.
\theoremstyle{definition}
\newtheorem*{thm:P1}{Theorem \ref{thm: PolyNSW}}
\begin{thm:P1}
	For additive valuations, given an $\alpha$-approximate allocation to maximum Nash social welfare,
	there exists a polynomial time computable envy-free allocation with transfers $(A,t)$ such that
	$$
	\frac{\NSW(A,t)}{\NSW(A^*)}\ \geq\ \frac{1}{2}\alpha\cdot e^{-1/e} 
	$$
	with a total transfer $\sum_{i}|t_i|$ of at most $2n^2$.
\end{thm:P1}

In sharp contrast, for utilitarian social welfare, the factor $\frac{1}{n}$ welfare threshold is tight.
To achieve any welfare guarantee greater than $\frac{1}{n}$ requires non-negligible transfer payments. 
Specifically, we show

\theoremstyle{definition}
\newtheorem*{cor:B1}{Corollary \ref{cor: budget}}
\begin{cor:B1}
For any $\alpha \in \left[\frac{1}{n},1\right]$, there exists an instance with additive valuations such that any envy-free allocation with transfers $(A,t)$ satisfying 
$\frac{\sw(A,t)}{\sw(A^*)}\geq\alpha$ requires a total transfer $\sum_{i\in N}|t_i|$ of at least $\frac{1}{4}\left(\alpha-\frac{1}{n}\right)^2 m$. 
\end{cor:B1}

In fact, there exist instances for which any EF$k$ allocation with $k=o(m)$ has a welfare guarantee of at most $\frac{1}{n} + o(1)$ (Lemma~\ref{imposs}). This implies that EF$k$ allocations cannot provide higher welfare with moderate transfers.

On the positive side, we can design algorithms to produce envy-free allocations with welfare guarantee $\alpha$
whose total transfer payment is comparable to the minimum amount possible, quantified in terms of
the maximum value $\max\limits_i v_i(A^*_i)$ any agent has in the welfare-maximizing allocation.
\theoremstyle{definition}
\newtheorem*{thm:Pos}{Theorem \ref{Pos}}
\begin{thm:Pos}\label{Pos}
	For additive valuations, for any $\alpha\in \left(0,1\right]$, there is a polynomial time computable envy-free allocation with transfers $(A,t)$ such that 
	$$\frac{\sw(A,t)}{\sw(A^*)}\geq \alpha$$ 
	with total transfer $\sum_{i\in N} |t_i|\leq n(\alpha\max\limits_i v_i(A^*_i)+2)$.
\end{thm:Pos}

\theoremstyle{definition}
\newtheorem*{thm:up}{Theorem \ref{thm: upGen}}
\begin{thm:up}
	For general valuations, for any $\alpha\in \left(0,\frac{1}{3}\right]$,  there is an envy-free allocation with transfers $(A,t)$ such 
	that $$\frac{\sw(A,t)}{\sw(A^*)}\geq \alpha \ \ \ $$ with total transfer $\sum_{i\in N} |t_i|\leq 2n^2\left(3\alpha \max_i v_i\left(A^*_i\right)+2\right)$.
\end{thm:up}

\subsection{Overiew of Paper}

In Section~\ref{sec:prelim}, we present our model of the fair division problem with transfers. Section~\ref{sec:rho} contains an 
exposition of the prior results in the literature that will be useful, along with our preliminary results on 
the $\rho$-mean welfare of envy-free allocations with transfers. In Section~\ref{sec:nsw}, we present our results on 
Nash social welfare, and in Section~\ref{sec:usw} we present our results on utilitarian social welfare.

%xxxxxxxxxxxx optional changes/additions xxxxxxxxxxxx
%
%- discussion on NSW construed as a fairness property, or as a halfway property between fairness and welfare. More 
%details on rho-mean welfare.
%
%- relationship between subsidies and transfers elaborated (subsidy is at most n times transfer, transfer is at most 2 times 
%subsidy, both are bounded by n times the maximum payment to/from an agent). Not really necessary.
%
%xxxxxxxxxxxxxxxxxxxxxxxxxxxxxxxxxxxx

%% file: prelim.tex
\section{The Model and Preliminaries} \label{sec:prelim}
Let $M=\{1,\cdots, m\}$ be a set of $m$ indivisible items and let $N=\{1,\cdots,n\}$ be a set of agents. 
Each agent $i$ has a \textit{valuation function} $v_i:2^{M}\rightarrow \Real$, where $v_i(\emptyset)=0$. We make the standard 
assumption that each valuation function is \textit{monotone}, satisfying 
$v_i(S)\leq v_i(T)$ whenever $S\subseteq T$. Additionally, following previous work on subsidies (see e.g. \cite{HS19,BDN20}), without loss of generality we uniformly scale the valuation functions by the same factor for each agent
so that the maximum marginal value of any item is at most $1$. %(that is, $\max_{i\in N,j\in M}\max_{S\subseteq{J\setminus\{j\}}}v_i(S\cup\{j\})-v_i(S) \leq 1$).
Besides general monotone valuations, we are also interested in well-known classes of valuation function, in particular, {\em additive} (linear) valuations where $v(S) = \sum_{g\in S} v(g)$ for each $S \subseteq M$,
%{\em submodular} (decreasing marginal) valuations where $v(S\cup T)+ v(S\cap T)\le  v(S) + v(T)$ for all $S,T \subseteq M$,
and {\em subadditive} (complement-free) valuations where $v(S\cup T)\le  v(S) + v(T)$ for all $S,T \subseteq M$. We use $[n]$ to denote the set $\{1,\cdots,n\}$.

\subsection{Fairness and Welfare}
An allocation $A=(A_1,A_2,\cdots,A_n)$ is a partition of the items into $n$ disjoint subsets, where $A_i$ is the 
set of items allocated to the agent $i$. Our aim is to obtain allocations that are both fair and of high welfare.
The concept of fairness we use is envy-freeness.
\begin{definition}
	An allocation $A= (A_1,\cdots, A_n)$ is \textit{envy-free} if for each $i,j\in N$
	$$
	v_i(A_i)\geq v_i(A_j) 
	$$
\end{definition}
In other words, an allocation is envy-free if each agent $i$ prefers its own bundle $A_i$ over 
any the bundle $A_j$ of any other agent $j$. If agent $i$ prefers the bundle of agent $j$ then we say $i$ envies $j$. 
Unfortunately, envy-free allocations do not always exist with indivisible item. This is evident 
even with two agents and one item, since the agent without an item will always envy the other. Moreover, 
even with two players and with identical additive valuations, determining whether an envy-free allocation 
exists is NP-complete. Consequently weaker notions of fairness have been introduced~\cite{Bud11}, most notably 
envy-freeness up to one good. 

\begin{definition}
	An allocation is \textit{envy-free up to one good} (EF1) if for each $i,j\in N$
	$$
	\ v_i(A_i)\geq v_i(A_j\setminus g) \text{ for some } g\in A_j  
	$$
\end{definition}
Rather than approximate fairness, however, our focus is on obtaining envy-freeness by adding one divisible item (money).
Thus we have an {\em allocation with payments}; in addition to the bundle $A_i$ of indivisible good, an agent~$i$ has a payment 
$p_i$.
\begin{definition}
	An allocation with payments $(A,p)$ is \textit{envy-free} if for each $i,j\in N$ 
	$$
	v_i(A_i)+p_i \geq v_i(A_j)+p_j 
	$$ 
\end{definition}
Furthermore, we say that an allocation $A$ is {\em envy-freeable} if there exist payments $p$ such that $(A,p)$ is envy-free.
An important fact is that, in contrast to envy-free allocations, envy-freeable allocations always exist for monotone valuations~\cite{HS19}.
There are two natural types of payment. First, we have \textit{subsidy payments} if $p_i\geq 0$.
Second, we have \textit{transfer payments} if $\sum_{i\in N}p_i=0$, To distinguish these, we denote a subsidy 
payment to agent~$i$ by $s_i$ and a transfer payment by $t_i$. We define the \textit{total transfer} of an allocation as the sum $\sum_i |t_i|$.

We measure the welfare of an allocation $A$ using the general concept of $\rho$-{\em mean welfare}: 
$$\operatorname{W}^\rho(A)=\left(\frac{1}{n}\sum_{i\in N}v_i(A_i)^\rho\right)^{\frac{1}{\rho}}$$
This class of welfare functions, introduced by \citet{ABK19}, encompasses a range of welfare functions including the two 
most important cases ($\rho\rightarrow0$ and $\rho=1$). The former corresponds to 
{\em Nash social welfare}, the geometric mean of the values of the agents, denoted by $\NSW(A)=\left(\prod_{i\in N} v_i(A_i)\right)^{\frac{1}{n}}$.
The latter corresponds to the arithmetic mean of the values of the agents but,
scaling by the number of agents, this is more commonly known as the {\em utilitarian social welfare} or simply {\em social welfare}, i.e. 
the sum of the values of the agents, denoted by $\sw(A)=\sum_{i\in N} v_i(A_i)$.
With transfer payments, our interest lies in utilities rather than simply valuations. In particular, the
$\rho$-mean welfare of an allocation with transfers $(A,t)$ is 
$$\operatorname{W}^\rho(A,t)=\left(\frac{1}{n}\sum_{i\in N} \left(v_i(A_i)+t_i\right)^\rho\right)^{\frac{1}{\rho}}$$

\subsection{Fair Division With Transfer Payments}
In this paper, we study the following question.
\mybox{
\begin{center} Is there an allocation with transfers that simultaneously satisfies the properties of\\ (i) envy-freeness, (ii)  high welfare, and (iii) a negligible total transfer?
\end{center} }
We have seen that envy-freeable allocations always exist. Thus, with transfer payments, we can obtain the property of envy-freeness.

The reader may ask whether transfers are necessary. Specifically, given the guaranteed existence of envy-freeable allocation,
can such allocations provide high welfare? The answer is {\sc no}. Without transfers, high welfare is impossible to ensure.
Even worse, no positive guarantee on welfare can be obtained without transfers.
This is true even for the case of additive valuations. To see this, consider the following simple example for
Nash social welfare.
\begin{example}\label{ex:bad-NSW}
Take two agents and two items $\{a,b\}$. Let the valuation functions be additive with $v_{1,a}=1, v_{1,b}=\frac{1}{2}$ for agent 1 and
$v_{2,a}=\frac{1}{2}, v_{2,b}=\epsilon$ for agent 2. Observe there are only two envy-freeable allocations: either agent 1 gets
both items or agent 1 gets item ${a}$ and agent 2 gets item ${b}$. For both these envy-freeable allocations the 
corresponding Nash social welfare is at most $\sqrt{\epsilon}$. In contrast, the optimal Nash social welfare is $\frac{1}{2}$
when agent 1 gets item ${b}$ and agent 2 gets item ${a}$.	
\end{example}
It follows that to find envy-free solutions with non-zero approximation guarantees for welfare we must have transfer payments. At the outset, if we restrict $\rho$ to be equal to 1, the result of \citet{HS19} implies that the allocation that maximizes utilitarian welfare can be made envy-free with transfer payments. However, we show that this allocation can require arbitrarily large transfers relative to the number of agents. The main point of concern in using transfer payments to achieve envy-freeness is that it may be difficult for the participants to include a substantial quantity of money in the system in order to implement this solution. Consequently, this creates a third requirement, i.e. to bound the total transfers. Thus the holy grail here is to obtain high welfare using only negligible transfers. Formally, we desire transfers whose sum (of absolute values) is independent of the number of items $m$. In particular, we want an allocation with transfers $(A,t)$ such that the welfare of $A$ is at least $\alpha$ times the welfare of the welfare-maximizing allocation $A^*$ (for some large $\alpha \in [0,1]$)
and $\sum_{i\in N} |t_i| = O(f(n))$ for some function~$f$. Specifically, the payments are negligible in the number of items (and thus in the
total welfare) as $m$ grows.

At first glance, this task seems impossible. If envy-freeable solutions cannot themselves ensure non-zero welfare guarantees,
how could negligible transfer payments then induce high welfare? Very surprisingly, this is possible for some important classes 
of valuation functions. However, it is indeed not always possible for other classes. Investigating how and where the boundary of this dichotomy
lies is the purpose of this paper.

%
%Denote $l(i,j)$ as the max weight path from $i$ to $j$. Following is the lower bound on total transfers:
%\begin{lemma}
%	Let $(A,t)$ be an envy-free allocation with transfers then transfer payment $t$ satisfy 
%	$$
%	t_i \geq \frac{1}{n}  \sum_{j\in N} l(i,j)
%	$$  
%\end{lemma}
%\begin{proof}Observe that 
%	\begin{align*}
%		t_i=t_i-\frac{1}{n}\sum_{j\in N} t_j &= \frac{1}{n}\sum_{j\in N} t_i-t_j
%	\end{align*}
%	Suppose the heaviest path from $i$ to $j$ is of the form $(\pi(1),\pi(2),\cdots ,\pi(r))$, with $\pi(1)=i$ and $\pi(r)=j$. Note that $v_{\pi(i)}(A_{\pi(i)})+t_{\pi(i)}\geq v_{\pi(i+1)}(A_{{\pi(i+1)}})+t_{\pi(i+1)}$ thus 
%\begin{align*}
%	t_i-t_j & = \sum_{i=1}^{r-1} t_{\pi(i)}-t_{\pi(i+1)} \\
%	&\geq \sum_{i=1}^{r-1}v_{\pi(i+1)}(A_{{\pi(i+1)}})-v_{\pi(i)}(A_{\pi(i)}) \\
%	& = l(i,j) 
%	\end{align*}	 
%\end{proof}

%% file: FairDivisionTransfer.tex
\section{Transfer Payments and $\rho$-Mean Welfare} \label{sec:rho}

In this section we familiarize the reader with the structure of envy-freeable allocations and transfer payments, and introduce 
our preliminary results. We begin with the general case of $\rho$-mean welfare. For subadditive valuations we have the following welfare guarantee.
\begin{lemma}\label{lem:rho-mean}
	For subadditive valuations, any envy-free allocation with transfers $(A,t)$ satisfies
	$$
	\operatorname{W}^\rho(A,t)\geq \frac{1}{n}\operatorname{W}^\rho(A^*)
	$$
\end{lemma}
\begin{proof}
	By the envy-freeness property $v_i(A_i)+t_i\geq  v_i(A_j)+t_j$. Thus 
	\begin{align*}
		v_i(A_i)+t_i \geq \frac{1}{n}\left( \sum_{j} v_i(A_j)+t_j\right) \geq \frac{1}{n}v_i(M)  \geq \frac{1}{n}v_i(A^*_i)
	\end{align*}
Here the second inequality follows by subadditivity. Hence 
\begin{align*}
\operatorname{W}^\rho (A,t) 
&\ = \ \left(\frac{1}{n}\sum_{i\in N}(v_i(A_i)+t_i)^\rho\right)^{\frac{1}{\rho}} \\ 
&\ \geq \  \left(\frac{1}{n}\sum_{i\in N}\left(\frac{v_i(A^*_i)}{n}\right)^\rho\right)^{\frac{1}{\rho}} \\
& \ = \ \frac{1}{n}\left(\frac{1}{n}\sum_{i\in N}\left(v_i(A^*_i)\right)^\rho\right)^{\frac{1}{\rho}} \\
& \ = \ \frac{1}{n}\operatorname{W}^\rho(A^*) 
\end{align*}
as desired.
\end{proof}
The resultant welfare guarantee of $\alpha=\frac{1}{n}$ is not particularly impressive.
But, at least, it is a strictly positive guarantee, which was unachievable without transfer payments.
The bound is also tight as shown by the following simple example. 
\begin{example}
Take $m=n$ items and $n$ agents. Let the valuation functions be additive with $v_{ii}=1$ and $v_{ij}=0$ for $j\neq i$.
Consider the allocation assigning the grand bundle to agent $1$. This is envy-freeable with transfer payments $t_1=-\frac{n-1}{n}$ and
$t_i=\frac{1}{n}$, for any agent $i\neq 1$. For social welfare ($\rho=1$) the corresponding welfare 
guarantee is $\alpha=\frac{1}{n}$.  
\end{example}

But how expensive is it to obtain this welfare guarantee?
To answer this, we provide a short review concerning the computation of transfer payments.
Recall that an allocation $A$ is envy-freeable if there exist payments $p$ such that $(A,p)$ is envy-free.  
Furthermore, there is a very useful graph characterization of envy-freeability. 
Given an allocation $A$ we build an envy-graph, denoted $G_{A}$.
The envy-graph is directed and complete. It contains a vertex for each agent $i\in N$.
For any pair of agents $i,j\in N$, the weight of arc $(i,j)$ in $G_{A}$
is the envy agent $i$ has for agent $j$ under the allocation $A$, that is, 
$w_{A}(i,j) =  v_i(A_{j}) - v_i(A_i)$. The envy-graph induces the following characterization.
\begin{theorem}[\cite{HS19}] \label{thm:hsthm1}
	The following statements are equivalent.
	\begin{enumerate}%[itemsep=3pt, topsep=3pt]
		\item[i)] The allocation $A$ is envy-freeable.
		\item[ii)] The allocation $A$ maximizes (utilitarian) welfare across all reassignments of its bundles to agents: for 
		every permutation $\pi$ of $N$, we have $\sum_{i\in N}v_i(A_i) \geq \sum_{i\in N}v_i(A_{\pi(i)})$.
		\item[iii)] The envy graph $G_{A}$ contains no positive-weight directed cycles.
	\end{enumerate}
\end{theorem}  
In addition, we can use the envy-graph to compute the transfer payments. First, it is known~\cite{HS19} how to
find, for any envy-freeable allocation $A$, the minimum {\em subsidy} payments $s$ such that $(A,s)$ is
envy-free. Specifically, let $l(i)$ be weight of a maximum weight path from node $i$ to any other node in $G_A$. Setting 
$s_i=l(i)$, for each agent $i$, gives an envy-free allocation with minimum subsidy payments. 
We do not wish to subsidize the mechanism, so we convert these subsidies into transfer payments.
To do this, let $\bar{s}=\frac{1}{n}\sum_{i\in N} s_i$ be the average subsidy. Then setting
$t_i=s_i-\bar{s}$ for each agent gives a valid set of transfer payments, which we dub the
\textit{natural} transfer payments. We remark that the natural transfer payments do not always minimize the total transfer, 
but they will be sufficient for our purposes.

We are now ready to compute transfer payments for subadditive valuations in the $\rho$-mean welfare setting. 
We begin with a theorem of \citet{BDN20}.
\begin{theorem}[\cite{BDN20}]\label{thm: EnvyMatch}
For monotone valuations there is a polynomial time algorithm to find an envy-free allocation 
with subsidies $(A,s)$ such that $s_i\leq 2(n-1)$ for all $i$.
\end{theorem}
Observe that when these subsidies are converted to the corresponding natural transfers, any bound on the maximum 
subsidy for each agent also applies to the maximum transfer for each agent. Combining this observation with the previous 
result gives us the following corollary.
\begin{corollary}\label{cor: EnvyMatch}
For subadditive valuations, there exists an envy-free allocation with transfers $(A,t)$ such that 
$$
\frac{\W^{\rho}(A,t)}{\W^{\rho}(A^*)} \ge \frac{1}{n} 
$$ 
and with a total transfer $\sum_{i}|t_i|$ of at most $2n^2$. This allocation can be computed in polynomial time.
\end{corollary}

Thus, we can quickly obtain an envy-free allocation with transfers whose total transfer is independent of the number of items $m$.
So we have negligible transfers. But, as stated, we only have a low welfare guarantee for this general $\rho$-mean welfare class.
In the next section, we will show that high welfare and negligible transfers are achievable for the special case of $\rho=0$, that is, Nash social welfare.

Before doing so, we conclude this section by presenting a generalization of Theorem~\ref{thm: EnvyMatch} that will later be useful. 
We say that an allocation $B$ has {\em b-bounded envy} if $v_i(B_j)-v_i(B_i)\leq b$ for every pair of agents $i,j\in N$.

\begin{lemma}\label{lem: EnvyMatch}
Given an allocation $B$ with $b$-bounded envy, there is a polynomial time algorithm to find an envy-free allocation with transfers $(A,t)$ such that $\sum_{i\in N}|t_i|\leq 2bn^2$. 
\end{lemma}

\begin{proof}
%By Lemma~\ref{lem:rho-mean}, it suffices to find an allocation with $(A,t)$ whose transfers sum to at most $2n^2$.
Let $B=\{B_1,B_2,\dots, B_n\}$ be an allocation with $b$-bounded envy. Let $A=(B_{\pi(1)},\cdots,B_{\pi(n)} )$ be the envy-freeable allocation obtained by computing a maximum-weight 
matching between the bundles in $B$ and the agents. Applying an approach of~\cite{BDN20}, let $P$ be 
a path of maximum weight in the envy-graph $G_{A}$.
Without loss of generality, $P=(1,\cdots, r)$.
By definition of the envy-graph, we then have
	\begin{align}\label{eq:rho-1}
		w_{A}(P)&=\sum_{i=1}^{r-1} v_i(A_{i+1})-v_i(A_{i}) \nonumber\\ 
		&=\sum_{i=1}^{r-1} v_i(A_{i+1})-v_i(B_i)+ v_i(B_i)- v_i(A_{i}) \nonumber\\ 
		&=\sum_{i=1}^{r-1} v_i(B_{\pi(i+1)})-v_i(B_i)+\sum_{i=1}^{r-1}  v_i(B_i)- v_i(A_{i}) \nonumber\\ 
		&\leq b(n-1)+\sum_{i=1}^{r-1}v_i(B_i)- v_i(A_{i})
	\end{align}
Here the inequality holds as $v_i(B_{\pi(i+1)})-v_i(B_i)\leq b$ for each agent $i$, and $r < n$. We have
\begin{align}\label{eq:rho-2}
\sum_{i=1}^{r-1}v_i(B_i)- v_i(A_{i}) &\leq
		\sum\limits_{i:v_i(B_i)\geq v_i(A_i) }v_i(B_i)- v_i(A_{i}) \nonumber\\ 
		&\leq -\sum\limits_{i:v_i(B_i)< v_i(A_i) }v_i(B_i)- v_i(A_{i}) \nonumber\\ 
		&=  \sum\limits_{i:v_i(B_i)< v_i(A_i) } v_i(A_{i})-v_i(B_i) \nonumber\\ 
		&=  \sum\limits_{i:v_i(B_i)< v_i(A_i) } v_i(B_{\pi(i)})-v_i(B_i) \nonumber\\ 
		&\leq b(n-1).
	\end{align}	
Above the second inequality holds as the social welfare of $A$ is the maximum over all allocations of the bundles in $B$; in particular, 
$\sum_i v_i(A_i) \ge \sum_i v_i(B_{i})$. The last inequality again follows as $B$ has $b$-bounded envy.

Together (\ref{eq:rho-1}) and (\ref{eq:rho-2}) give $w_{A}(P)\leq 2b(n-1)$.
This implies that $A$ can be made envy-free with a subsidy $s_i\leq 2b(n-1)$ to each agent $i$. 
Hence, setting $t_i=s_i-\bar{s}$, we have that $(A,t)$ is 
	envy-free with a total transfer payment of at most $\sum_{i\in N}|t_i|\leq 2bn^2$.
\end{proof}

%\begin{proof}
%Let $A=(B_{\pi(1)},\cdots,B_{\pi(n)} )$ be the envy-freeable allocation obtained by computing maximum-weight 
%matching between the bundles in $B$ and the agents. Now let $P$ be the largest weight path in the envy graph $G_{A}$.
%Without loss of generality, $P=(1,\cdots, r)$.
%We then have
%	\begin{align*}
%		w_{A}(P)&=\sum_{i=1}^{r-1} v_i(A_{i+1})-v_i(A_{i}) \\ 
%		&=\sum_{i=1}^{r-1} v_i(A_{i+1})-v_i(B_i)+ v_i(B_i)- v_i(A_{i}) \\ 
%		&\leq (n-1)b+\sum_{i=1}^{r-1}v_i(B_i)- v_i(A_{i})
%	\end{align*}
%Where the last inequality follows as $v_i(A_{i+1})-v_i(B_i)=v_i(B_{\pi(i+1)})-v_i(B_i)\leq b$.
%	Note that $\sum_{i=1}^{r-1}v_i(B_i)- v_i(A_{i})\leq \sum\limits_{i:v_i(B_i)\geq v_i(A_i) }v_i(B_i)- v_i(A_{i}) $.
%	Also since by definition $\sw(B)\leq \sw(A) $ , we see that
%	\begin{align*}
%		\sum\limits_{i:v_i(B_i)\geq v_i(A_i) }v_i(B_i)- v_i(A_{i})&\leq -\sum\limits_{i:v_i(B_i)< v_i(A_i) }v_i(B_i)- v_i(A_{i}) \\
%		&=  \sum\limits_{i:v_i(B_i)< v_i(A_i) } v_i(A_{i})-v_i(B_i) \\
%		&\leq (n-1)b
%	\end{align*}
%	Thus indeed $w_{A}(P)\leq 2b(n-1)$, which implies $s_i\leq 2b(n-1)$. Hence setting $t_i=s_i-\bar{s}$ we see that $(A,t)$ is 
%	envy-free, and $\sum_{i\in N}|t_i|\leq 2bn^2$.
%\end{proof}

%{\bf [Put back in this proof for b-bounded valuations and remove the proof of Theorem~\ref{thm: EnvyMatch} as it copies the proof in other paper....]}

%% file: NSW.tex
\section{Transfer Payments and Nash Social Welfare}\label{sec:nsw}
In the following two sections, we present our main results concerning Nash social welfare and utilitarian social welfare.
Here we show that, with transfers, excellent welfare guarantees can be obtained for Nash social welfare.
Conversely, in Section~\ref{sec:usw}, we will see that only much weaker guarantees can be obtained for utilitarian social welfare.

\subsection{NSW with General Valuation Functions}
Now, recall from Example~\ref{ex:bad-NSW} that no positive welfare guarantee can be obtained in the case of Nash social welfare
even for the basic case of additive valuations. Our first result for Nash social welfare is therefore somewhat surprising.
With transfer payments, constant factor welfare guarantees can be obtained for NSW for general valuations.
That is, envy-freeness and high welfare are simultaneously achievable. 
\begin{theorem}\label{thm: NSW}
	For general valuations, there exists an envy-free allocation with transfers $(A,t)$ such that 
	$$
	\frac{\NSW(A,t)}{\NSW(A^*)}\geq e^{-1/e} 
	$$
\end{theorem}
\begin{proof} 
Let $A^*$ be an allocation that maximizes Nash social welfare. Now, let
$A$ be an envy-freeable allocation induced by reallocating the bundles in $A^*$ to maximize utilitarian social welfare.
Recall this can be found by taking a maximum weight matching between the agents and
the bundles of $A^*$; let $\pi(i)$ be the agent who receives bundle $A^*_i$ in the allocation $A$.
By Theorem~\ref{thm:hsthm1}, this allocation is envy-freeable. 
So let $t$ be any valid set of  transfer payments such that $(A,t)$ is envy-free.

By definition we have that $v_i(A^*_i)=v_i(A_{\pi(i)})$, for all $i\in N$.
Then, by envy-freeness,  we have  $v_i(A_i)+t_i\geq v_i(A_{\pi(i)})+t_{\pi(i)}=v_i(A^*_i)+t_{\pi(i)}$. Denote by $t_{\max}$ the 
maximum positive transfer payment, i.e. $t_{\max} = \max_{i} t_i$, and let $m$ be an agent whose transfer $t_m$ is 
equal to $t_{\max}$. By envy-freeness, no agent envies agent $m$, so $v_i(A_i)+t_i \geq t_{\max}$ for all $i$. Putting this all together, we have
\begin{align*}
	\frac{\prod_{i=1}^{n} v_i(A_i)+t_i }{\prod_{i=1}^{n}v_i(A^*_i)} 	 	 
	&\ \geq\ \prod_{i=1}^n \frac{\max\left[ v_i(A^*_i)+t_{\pi(i)} , t_{\max}\right] }{v_i(A^*_i)}
\end{align*}
%To see this, observe that the maximum payment $t_{\max}$ is non-negative as transfer payments sum to zero.
%Thus, envy-freeness implies that an each must do at least as well as receiving the maximum transfer payment plus it associated
%bundle (of non-negative value).

Now define $N^+=\{i \ | \ t_{\pi(i)}\geq 0 \} $ and $N^- = N\setminus N^+$. 
\begin{align*}
	\prod_{i=1}^n \frac{\max\left[ v_i(A^*_i)+t_{\pi(i)} , t_{\max}\right] }{v_i(A^*_i)} 
	&\geq \prod_{i\in N^+} \frac{v_i(A^*_i)+t_{\pi(i)} }{v_i(A^*_i)} \cdot \prod_{i\in N^-} \frac{\max\left[ v_i(A^*_i)+t_{\pi(i)} , t_{\max}\right] }{v_i(A^*_i)}\\
	&\ge \prod_{i\in N^-} \frac{\max\left[ v_i(A^*_i)+t_{\pi(i)} , t_{\max}\right] }{v_i(A^*_i)}
\end{align*}
Next let $N^-_1$ be the indices corresponding to negative transfers that also satisfy $t_{\max}\le v_i(A^*_i)+t_{\pi(i)}$, and 
let $N^-_2$ be the indices corresponding to negative transfers that also satisfy $t_{\max} > v_i(A^*_i)+t_{\pi(i)}$. 
Furthermore, set $v_i(A_i^*)+t_{\pi(i)} = t_{\max}+\alpha_i$. Observe that, for $i\in N^-_1$, we have $\alpha_i\ge 0$, but 
for $i\in N^-_2$, we have $\alpha_i< 0$. Applying this gives
\begin{align*}
	\prod_{i\in N^-} \frac{\max\left[ v_i(A^*_i)+t_{\pi(i)} , t_{\max}\right] }{v_i(A^*_i)}
	&\ \ge\  \left(\prod_{i\in N^-_1} \frac{t_{\max}+\alpha_i }{t_{\max}+\alpha_i-t_{\pi(i)}}\right)\cdot \left(\prod_{i\in N^-_2} \frac{t_{\max}}{t_{\max}+\alpha_i-t_{\pi(i)}}\right)\\
	&\ \ge\  \left(\prod_{i\in N^-_1} \frac{t_{\max}}{t_{\max}-t_{\pi(i)}}\right)\cdot \left(\prod_{i\in N^-_2} \frac{t_{\max}}{t_{\max}-|\alpha_i|-t_{\pi(i)}}\right)\\
	&\ \ge\  \left(\prod_{i\in N^-_1} \frac{t_{\max}}{t_{\max}-t_{\pi(i)}}\right)\cdot \left(\prod_{i\in N^-_2} \frac{t_{\max}}{t_{\max}-t_{\pi(i)}}\right)\\
	&\ =\  \left(\prod_{i\in N^-} \frac{t_{\max}}{t_{\max}-t_{\pi(i)}}\right)	
\end{align*}
Now for, $i\in N^-$, let $k_i= |t_{\pi(i)}|$.
Since $\sum_{i\in N} t_i=0$ we have $\sum_{i\in N^+}t_i=\sum_{i\in N^-}t_i := T$.   Thus
\begin{align*}
	\left(\frac{\prod_{i=1}^{n} v_i(A_i)+t_i }{\prod_{i=1}^{n}v_i(A^*_i)} \right)^{1/n}
	\ &\ge\  	\left(\prod_{i\in N^-} \frac{t_{\max}}{t_{\max}-t_{\pi(i)}}\right)^{1/n} \\
	 \ &=\  \left(\prod_{i\in N^-} \frac{t_{\max}}{t_{\max}+k_i}\right)^{1/n}  
\end{align*}
Observe, by the {\em arithmetic-geometric mean inequality}, that  $\prod_{i\in N^-}(t_{\max}+k_i) $ is maximized when $k_i=k_j=T/|N^-|$.
In addition, $t_{\max}\geq T/|N^+| $. So
\begin{align*}
	\left(\prod_{i\in N^-} \frac{t_{\max}}{t_{\max}+k_i}\right)^{1/n} 
	\ &\geq\ \left(\frac{\frac{T}{|N^+|}}{\frac{T}{|N^+|}+\frac{T}{|N^-|}}\right)^{|N^-|/n}  \\
	\ &=\  \left(\frac{n-|N^+|}{n}\right)^{\frac{n-|N^+|}{n}} \\
	\ &\geq\ \min\limits_x \left(\frac{1}{x}\right)^\frac{1}{x} \\ 
	\ &\geq\ e^{-1/e} &  \qedhere
\end{align*}
\end{proof}

This theorem is rather noteworthy; for general valuation functions, with transfers, it allows us to simultaneously obtain both high Nash social welfare and perfect envy-freeness. But what of our third objective, that of negligible transfer payments? The approach applied in the proof of Theorem~\ref{thm: NSW} cannot guarantee negligible transfers. Specifically, simply reallocating the bundles of the allocation $A^*$ that maximizes Nash social welfare can require large transfers. In particular, the following example shows this method may require transfers as large as $\Omega(\sqrt{m})$.
\begin{example}\label{ex:reallocate-bundles}
Take an instance with two agents and $m$ items. Assume the first agent has a valuation function given by 
$v_1(S)=|S|$, for each $S\subseteq M$; assume the second agent has a valuation function given by  $v_2(S)=\sqrt{|S|}$, 
for each $S\subseteq M$. The reader may verify that the Nash welfare maximizing allocation $A^*$ is to 
give the first agent $\frac{2m}{3}$ items and the second agent $\frac{m}{3}$ items. 
This allocation is also the allocation that maximizes utilitarian social welfare by reassigning the bundles of $A^*$. 
Thus $A=A^*$. However, to make the allocation envy-free requires a minimum transfer payment of $\Omega(\sqrt{m})$, 
from the first agent to the second agent.
\end{example}  
Of course, this example does not rule out the possibility that, for general valuation functions,
an envy-free allocation with transfers that has high welfare and negligible 
payments exists. In particular, simply allocating each agent half the items requires no transfer payments at all, and gives high Nash social welfare. So simultaneously obtaining high Nash social welfare and envy-freeness via negligible transfers for general valuation functions remains an open question. %There is a additional motivating factor for providing a bound on the total transfers. The lack of this type of guarantee necessitates that an agent participating in such a mechanism has sufficient capital beforehand in order to cover any transfer payments that must be made.
Fortunately, we can show that these three properties are simultaneously achievable for important special classes of valuation function.

\subsection{NSW Guarantees with Negligible Transfers}
Here we prove that for (i) additive valuations, and (ii) matroid rank valuations, it is always possible to
obtain envy-free allocations with high Nash social welfare and negligible transfers.
Furthermore, for additive valuations we can do this using polynomial time algorithms.
\begin{theorem}\label{thm: PolyNSW}
	For additive valuations, given an $\alpha$-approximate allocation to maximum Nash social welfare,
	there exists a polynomial time computable envy-free allocation with transfers $(A,t)$ such that
	$$
	\frac{\NSW(A,t)}{\NSW(A^*)}\ \geq\ \frac{1}{2}\alpha\cdot e^{-1/e} 
	$$
	with a total transfer $\sum_{i}|t_i|$ of at most $2n^2$.
\end{theorem}
\begin{proof}
Let $B$ be the $\alpha$-approximate allocation to the maximum Nash social welfare; that is $\frac{\NSW(B)}{\NSW(A^*)}\geq \alpha$.
Now Caragiannis et al~\cite{CGH19} gave a polytime algorithm which, given input $B$, outputs an 
EF1 allocation $B'$ with a Nash social welfare guarantee of $\frac{\alpha}{2}$.    
	
Next, recall the proof of~Theorem~\ref{thm: NSW}. Observe that, during the proof, we did not use the fact that $A^*$ maximizes 
Nash social welfare. Thus the $e^{-1/e}$ approximation ratio holds if we start with any other allocation $\hat{A}$ instead of $A^*$.
That is by reallocation the bundles of $\hat{A}$ we obtain an envy-freeable allocation $A$ whose Nash social welfare is that least
a factor $e^{-1/e}$ of that of $A^*$. In particular, we can do this for the allocation $\hat{A}=B'$ given by Caragiannis et al~\cite{CGH19}.
So, by Theorem~\ref{thm: NSW}, there exists an envy-free allocation with transfers $(A,t)$ such that $\frac{\NSW(A,t)}{\NSW(B')}\geq e^{-1/e}$. Now
\begin{align*}
		\frac{\NSW(A,t)}{\NSW(A^*)}
		\ =\ \frac{\NSW(B')}{\NSW(A^*)}\cdot \frac{\NSW(A,t)}{\NSW(B')} 
		\ \geq\  \frac{1}{2}\alpha \cdot e^{-1/e}
\end{align*}
Furthermore, because $B'$ is EF1 and $A$ is obtained by the same procedure as in Theorem~\ref{thm: EnvyMatch}, we obtain
transfer payments with $\sum_{i}|t_i|\leq 2n^2$.  
\end{proof}
We remark that, for additive valuations, polytime algorithms do exist to find allocations that $\alpha$-approximate the maximum Nash social welfare.
Specifically, Barman et al.~\cite{BKV18} present an algorithm with an approximation guarantee of $\alpha=\frac{1}{1.45}$.
Together with Theorem~\ref{thm: PolyNSW}, we thus obtain in polytime an envy-free allocation with negligible transfers and a Nash social welfare 
guarentee of $\frac{1}{2.9}e^{-1/e}$.

Better existence bounds can be obtained for the additive case if we remove the requirement of a polynomial time algorithm. 
A well-known result of~\citet{CKM19} states that for additive valuations, the Nash welfare maximizing allocation is EF1. In fact, 
a recent result of~\citet{BCI20} provides a similar result for the case of \textit{matroid rank} valuation functions, a sub-class of 
submodular functions. A valuation function is matroid rank if it is submodular, and the marginal value of any item is 
binary (i.e. for any set $S$ of items and any item $x$ not in $S$, $v_i(S\cup\{x\}) - v_i(S) \in \{0,1\}$). Here,
a NSW-maximizing allocation is EF1\cite{BCI20}. Combining this with Lemma~\ref{lem: EnvyMatch}, 
the corresponding envy-free allocation with transfers $(A,t)$ has transfers satisfying $\sum_{i}|t_i|\leq 2n^2$. Further, by Theorem~\ref{thm: NSW}, 
we have $\frac{\NSW(A,t)}{\NSW(A^*)}\geq e^{-1/e}$ as desired. 

\begin{theorem}\label{thm: matroid}
For matroid rank valuations, there exists an envy-free allocation with transfers $(A,t)$ with
$\frac{\NSW(A,t)}{\NSW(A^*)}\geq e^{-1/e}$ and $\sum_{i}|t_i|\leq 2n^2$. \qed
\end{theorem}

%% file: TPSW.tex
\section{Transfer Payments and Social Welfare}\label{sec:usw}

In this section we present our results on utilitarian social welfare.

\subsection{The Necessity of Non-Negligible Transfer Payments}
To begin, recall that an allocation $B$ has {\em b-bounded envy} if $v_i(B_j)-v_i(B_i)\leq b$ for every pair of agents $i,j\in N$.
Without transfers, allocations with  {\em b-bounded envy} may have very low welfare.
\begin{lemma}\label{imposs}
For utilitarian social welfare, there exist instances with additive valuation functions such that any allocation with $b$-bounded envy 
has a welfare guarantee of at most $2\sqrt{\frac{b}{m}}+\frac{1}{n}$. 	
\end{lemma}
\begin{proof}
		 Consider the following instance with additive valuations. Let $v_{n,j}=1$ for each $j\in M$ and let 
		 $v_{ij}=\epsilon$ for all $i\neq n$ and all $j\in M$. Evidently, to maximize utilitarian social welfare we simply give all the items to
		  agent $n$. So $\sw(A^*)=m$. Because the items are interchangeable for every agent, any allocation $A$ can be described 
		  as $(y_1,\cdots,y_{n-1},y_n=x)$, where $y_i$ is the fraction of items allocated to agent $i$.
		  to the agent. Since every item must be allocated, we have $\sum_{i=1}^{n-1} y_i = (1-x)$. 
		  The corresponding welfare guarantee for the allocation $A$ is then $\frac{\sw(A)}{\sw(A^*)} = (1-\epsilon)x+\epsilon$. 
		 
		 Now suppose $A$ has $b$-bounded envy. Therefore, $v_i(A_j)-v_i(A_i)\leq b$, for any pair of agents $i,j\in N$.
		 In particular, $m(y_i-x)\leq b$ since agent $n$ cannot envy agent $i$ too much and
		  $\epsilon m(x-y_i)\leq b$ since agent $i$ cannot envy agent $n$ too much. 
		  Summing the later inequality over all agents $i$ gives $\epsilon m((n-1)x-(1-x))\leq (n-1)b$. This 
		 implies $x\leq (1-\frac{1}{n})\frac{b}{\epsilon m}+\frac{1}{n}$. Thus 
		 \begin{align*}
		 	\frac{\sw(A)}{\sw(A^*)} &=(1-\epsilon)x+\epsilon \\
		 	&\leq  (1-\epsilon)\left(\left(1-\frac{1}{n}\right) \frac{b}{\epsilon m}+\frac{1}{n}\right) +\epsilon \\
		 	&\leq 2\sqrt{\frac{b}{m}}-\frac{b}{m}+\frac{1}{n}\\
			&\leq 2\sqrt{\frac{b}{m}}+\frac{1}{n}
		 \end{align*}
		 Here the second inequality holds by setting $\epsilon=\sqrt{\frac{b}{m}}$.  		 
	\end{proof}

Lemma~\ref{imposs} implies that any EF$k$ allocation in the given example, with $k=o(m)$, cannot provide a welfare guarantee 
that is significantly higher than $\frac{1}{n}$. The natural question to ask, now, is whether the problem inherent in Lemma~\ref{imposs} 
can be rectified with a small quantity of transfers. On the positive side, the result of~\citet{BDN20} shows that a small quantity of 
subsidy independent of the number of items is always sufficient to eliminate envy. 
A similar result also extends to the corresponding natural transfer payments. Combining this result with Lemma~\ref{lem:rho-mean} tells us 
that a utilitarian welfare guarantee of $\frac{1}{n}$ can be achieved alongside envy-freeness with a negligible total transfer. Unfortunately, for 
the above example, the Iterated Matching Algorithm of~\cite{BDN20} returns an allocation whose social welfare is only a $\frac{1}{n}$-fraction 
of the optimal welfare. The following corollary shows that this was inevitable: unlike for NSW, in order to make any improvement above 
this threshold, non-negligible transfers are required.

%Unfortunately, unlike for NSW, non-negligible transfers are required to rectify the problem 
%inherent in Lemma~\ref{imposs}. Specifically, to ensure high utilitarian social welfare, large transfers are necessary even 
%for additive valuation functions.

\begin{corollary}\label{cor: budget}
For any $\alpha \in \left[\frac{1}{n},1\right]$, there exists an instance with additive valuations such that any envy-free allocation with transfers $(A,t)$ satisfying 
$\frac{\sw(A,t)}{\sw(A^*)}\geq\alpha$ requires a total transfer $\sum_{i\in N} |t_i|\geq \frac{1}{4}\left(\alpha-\frac{1}{n}\right)^2 m$. 
\end{corollary}
\begin{proof}
	Take the same instance as in Lemma~\ref{imposs}. Now for utilitarian social welfare, we have $\sw(A)=\sw(A,t)$ 
	as $\sum_{i\in N} \left(v_i(A_i)+t_i \right) = 
	\sum_{i\in N} v_i(A_i)+ \sum_{i\in N} t_i=\sum_{i\in N} v_i(A_i)$. Let $A^*$ be the allocation that maximizes the social welfare.
	Thus 
	 \begin{align} \label{uswlabel1}
	\frac{\sw(A,t)}{\sw(A^*)} \ =\ (1-\epsilon)x+\epsilon \ \geq\ \alpha
	 \end{align}	
	Next, observe that $x=y_n\geq y_i$, for each $i$ otherwise the allocation $A$ is not envy-freeable. Thus $t_n \leq 0$. Then, by envy-freeness of $(A,t)$,
	we must have $mx +t_n\geq my_i+t_i $ and $\epsilon m y_i +t_i \geq \epsilon m x +t_n $. It follows that
	 \begin{align*}
	 (n-1)\cdot\left(\epsilon m x +t_n\right) 
	 \ \le \ \epsilon m\cdot \sum_{i=1}^{n-1}y_i + \sum_{i=1}^{n-1}t_i
	\ = \ \epsilon m\cdot (1-x) - t_n
	 \end{align*}	
Rearranging we obtain $n\cdot\left(\epsilon m x +t_n\right) \leq \epsilon m$. In particular,
	 \begin{align} \label{uswlabel2}
	  -t_n \ \geq\  \epsilon m\cdot \left(x-\frac{1}{n}\right) 
	 \end{align}
Combining (\ref{uswlabel1}) and (\ref{uswlabel2}) we get
	 \begin{align*}
	 \sum_{i\in N} |t_i|
	 \ \ge \ |t_n|
	 \ \geq\  -t_n
	  \ \geq\  \epsilon m\cdot\left(x-\frac{1}{n}\right) 
	 \ \geq\ \epsilon m\cdot\left(\frac{\alpha-\epsilon}{1-\epsilon}-\frac{1}{n}\right)
	  \end{align*}
Finally, choosing $\epsilon=1-\sqrt{\frac{1-\alpha}{1-\frac{1}{n}}}$ gives the desired bound. 
\end{proof}

\subsection{Constant-Sum Valuations}
So, for utilitarian social welfare, non-negligible transfers are required to ensure both envy-freeness and high welfare.
Recall, though, that balancing constraints on the valuation functions have been used in the literature to circumvent 
impossibility bounds on welfare. The reader may wonder if such constraints could be used to bypass the result in Corollary~\ref{cor: budget}: are negligible transfer payments sufficient to obtain high welfare when the valuation functions are constant-sum?
The answer is {\sc no}, as we shall see in the subsequent theorem.

In recent work, \citet{BBS20} considered the case of subadditive valuations with the constant-sum condition, and gave a polynomial-time algorithm that finds an EF1 allocation with social welfare at least $\Omega(\frac{1}{\sqrt{n}})$ of the optimal welfare. Applying the algorithm of Lemma~\ref{lem: EnvyMatch} to the resulting allocation gives us an envy-free allocation with negligible transfers and welfare ratio $\Omega(\frac{1}{\sqrt{n}})$. Once again, we show that this threshold cannot be crossed without non-negligible transfers.
\begin{theorem}	
There exist instances with constant-sum additive valuations such that any envy-free 
	allocation with transfers $(A,t)$ satisfying $\frac{\sw(A,t)}{\sw(A^*)}\geq\alpha$  
	has a total transfer $\sum_{i\in N} |t_i| \geq (\alpha-\frac{2}{\sqrt{n}}) \frac{m}{\sqrt{n}}$, for any $\alpha\in [\frac{2}{\sqrt{n}},1]$.%\Omega(\alpha \frac{m}{\sqrt{n}})$, for any $\alpha\in (\frac{2}{\sqrt{n}},1)$.
\end{theorem}
\begin{proof}
	Consider an instance with $m$ items and $n$ agents. Divide the items into $\sqrt{n}$ sets, 
	each of cardinality $\frac{m}{\sqrt{n}}$. Let $B_{\ell}$ be the set of items  $\{(\ell-1)\frac{m}{\sqrt{n}}+1, (\ell-1)\frac{m}{\sqrt{n}}+2,\cdots, \ell \frac{m}{\sqrt{n}}\}$. We now define a collection of constant-sum additive valuation functions. We partition the set of agents into two parts; agents in the set $H=\{1,\cdots, \sqrt{n}\}$ have high value for a small number of items, and agents in the set $L=\{\sqrt{n}+1,\cdots,n\}$ have low value for a large number of items.
	A high value agent $i$ has valuations $v_{ij}=1$ for $j\in B_i$ and zero otherwise. Thus for each agent $i\in H$ there is a corresponding set $B_i$ which it values. Each low value agent has a uniform valuation 
	of $v_{ij}=\frac{1}{\sqrt{n}}$ for all $j\in M$. 
	Observe that the value each agent has for the grand bundle is exactly $\frac{m}{\sqrt{n}}$, that is, constant-sum.	
	Note that any allocation to a high value agent $i$ can be described by the fraction 
	of $B_i$ which it receives. Consider an envy-freeable allocation $A$ that assigns an $x_i$-fraction and a $y_{ki}$-fraction of $B_i$ to $i\in H$ and $k\in L$ 
	respectively. By envy-freeability we must have $x_i \geq y_{ki}$ for all $i\in H$ and $k\in L$. We also have 
	that $x_i+\sum_{k\in L}y_{ki}=1$. Observe that the utilitarian social welfare is maximized by allocating $B_i$ to the high value agent $i$; this allocation satisfies $\sw(A^*)=m$. We then have
	\begin{align*}
		\frac{\sw(A)}{\sw(A^*)} 
		&\ =\ \frac{1}{m}\left(\frac{m}{\sqrt{n}}\sum_{i\in H}x_i+\frac{m}{\sqrt{n}}\sum_{i\in H}\sum_{k\in L}\frac{y_{ki}}{\sqrt{n}}\right) \\ 
		&\ =\ \frac{1}{m}\left(\frac{m}{\sqrt{n}}\sum_{i\in H}x_i+\frac{m}{n}\sum_{i\in H}(1-x_i)\right)\\ 
		&\ =\ \frac{1}{\sqrt{n}}\left(\left(1-\frac{1}{\sqrt{n}}\right)\sum_{i\in H}x_i+1\right)\\
		 &\ \geq\ \alpha
	\end{align*}
From this we can infer that $\sum_{i\in H}x_i \geq \sqrt{n}\alpha -1$. Now, let $t$ be valid transfer payments. 
By envy-freeness, we see that for any $i,k$ 
\[
\frac{m}{\sqrt{n}}\cdot\sum_{j\in H}\frac{y_{kj}}{\sqrt{n}}+t_k\geq \frac{m}{\sqrt{n}}\cdot \frac{x_i}{\sqrt{n}}+t_i 
\] 
First summing over $i\in H$ and then summing over $k\in L$ gives
\begin{align*}
%\frac{m}{\sqrt{n}}\sum_{j\in H}{y_{kj}}+\sqrt{n}t_k&\geq \frac{m}{n}\sum_{i\in H}x_i + \sum_{i\in H}t_i \\ 
\frac{m}{\sqrt{n}}\sum_{k\in L}\sum_{j\in H}{y_{kj}}+\sqrt{n}\sum_{k\in L}t_k 
\ \geq\  m\left(1-\frac{1}{\sqrt{n}}\right)\sum_{i\in H}x_i+(n-\sqrt{n})\sum_{i\in H}t_i 
\end{align*}
Sum $\sum_{i\in N}t_i=0$,  rearranging gives
\begin{align*}
0 \ \geq\  m\sum_{i\in H}x_i-    \frac{m}{\sqrt{n}}  \sum_{i\in H} \left( x_i+\sum_{k\in L}{y_{ki}}\right)+ n\sum_{i\in H}t_i 
\ \geq\  m\left(\sum_{i\in H}x_i-   1\right)+ n\sum_{i\in H}t_i 
\end{align*}
In particular,
\begin{align*}
-\sum_{i\in H}t_i \ \geq\  \frac{m}{n} \left(\sum_{i\in H}x_i-1\right) 
\end{align*}
Now recall that $\sum_{i\in H}x_i \geq \sqrt{n}\alpha -1$. Thus
\begin{align*}
\sum_{i\in N}|t_i| 
\ \geq\ -\sum_{i\in H}t_i 
\ \geq\  \frac{m}{\sqrt{n}}\left(\alpha-\frac{2}{\sqrt{n}}\right) 
\ \geq\  0   &\qedhere
\end{align*}

\end{proof}

So non-negligible transfer payments are required even assuming constant-sum valuations.
This adds to our collection of negative results for utilitarian social welfare. 
Are any positive results possible? Specifically, can we at least match the lower bounds on transfer payments 
inherent in the these negative results. We will now show this can indeed be approximately achieved.

\subsection{Upper Bounds on Transfer Payments}

To conclude the paper, we present results that upper bound the total transfer required to obtain an envy-free allocation
with a utilitarian social welfare guarantee. We give upper bounds for additive and general valuation functions. In both cases, the bound we obtain is expressed as a function of the maximum value that an agent receives in the welfare-optimal allocation. In particular, while the lower bounds are obtained as functions of $m$, the upper bounds we get are expressed as functions of the product of $n$ and $\max_i v_i(A^*_i)$. In allocations that distribute utility uniformly amongst the agents, these expressions are comparable; even in the worst case, since $v_i(A^*_i)\leq m$ for any $i$, they differ by some function of only $n$, and this difference is independent of the number of items. We begin with the additive case.
 
%{\bf [Discuss the switch from $m$ to $\max\limits_i v_i(A^*_i)$....]}

\begin{theorem}\label{Pos}
	For additive valuations, for any $\alpha\in \left(0,1\right]$, there is an envy-free allocation with transfers $(A,t)$ such that $$\frac{\sw(A,t)}{\sw(A^*)}\geq \alpha$$ 
	with total transfer $\sum_{i\in N}|t_i|\leq n(\alpha\max\limits_i v_i(A^*_i) +2)$.
\end{theorem}
\begin{proof}
We prove this result with a simple polytime algorithm (see Algorithm~\ref{alg: addUpper}) that outputs the desired allocation with transfers $(A,t)$.
\begin{algorithm}[ht]
	\SetAlgoLined
	$A_i\gets \emptyset$ for all $i\in N$\;
	$A^*=(A^*_1,\cdots, A^*_n)\gets $ Welfare-Maximizing Allocation\;
	\For{$i=1$ to $n$}{
		$A_i\gets$ minimal set $X_i\subseteq A^*_i$ with $v_i(X_i)\geq \alpha \cdot v_i(A^*_i)$	
	}	
	Use the Iterated Matching Algorithm of \cite{BDN20} to allocate $M\setminus \bigcup_{i\in N}A_i$\;
	Compute the natural transfers $(t_1,\cdots,t_n)$\;
	\Return{$(A,t)$}
	\caption{Envy-free allocation with high welfare and small transfers for additive valuations}\label{alg: addUpper}
\end{algorithm}

By additivity, the optimal allocation $A^*$ assigns each item in $M$ to an agent with the greatest valuation for that item. Consequently, $X=(X_1,\cdots,X_n)$ maximizes welfare among all reassignments of its bundles, so $X$ is an envy-freeable allocation of the items $\bigcup_{i\in [n]}X_i$. By construction, we have $\operatorname{SW}(X)\geq \alpha \operatorname{SW}(A^*)$. Now let $P$ be any path in the envy-graph $G_X$, without loss of generality, $P=(1,2,\dots, r)$. Then 
	\begin{align*}
		w_X(P) &= \sum_{i=1}^{r-1} v_i(X_{i+1})-v_i(X_{i}) \\ 
		&\leq - \left(v_r(X_1)-v_r(X_r)\right) \\
		&\leq \max\limits_i v_i(X_i) \\
		&\leq\alpha\max\limits_i  v_i(A^*_i) +1
	\end{align*}
	Here the first inequality holds by Theorem~\ref{thm:hsthm1} as the envy-graph contains no positive cycle.
	The last inequality holds by the minimality of $X_i$. Next, let $(Y_1,\cdots,Y_n)$ be the allocation of the 
	remaining items $\bigcup\limits_i (A_i^* \setminus X_i)$ given by the Iterated Matching Algorithm. The key properties we require from this algorithm are that the allocation $Y$ is envy-freeable and that, for any path $P$, weight of the path $w_Y(P)\leq 1$ (see \cite{BDN20}). But Algorithm~\ref{alg: addUpper} simply outputs the allocation $A_i=X_i\cup Y_i$ for each agent $i$. Hence  
	 	\begin{align*}
		w_A(P) \ =\ w_X(P)+w_Y(P) \ \leq\ \alpha\max\limits_i  v_i(A^*_i) +2
		\end{align*}
	Now if we take $s$ to be the minimum subsidy payments required for envy-freeness then $s_i\leq \alpha\max\limits_i  v_i(A^*_i) +2$, for each agent $i$.
	Using the transfer payments $t_i=s_i-\bar{s}$, we have that $\sum_{i}|t_i|\leq n(\alpha\max\limits_i v_i(A^*_i) +2)$,
	 as claimed.  
\end{proof}

%==============================================================
%
%{\bf [Where to put the following? Where it is first used?...]}
%In \cite{BDN20}, they show that for additive valuations at most $n-1$ dollars suffice to find an envy-free allocation with subsidy.
%\begin{theorem}[\cite{BDN20}]\label{thm:main}
%	For additive valuations, there is polynomial time computable  envy-free allocation with payments $(A,s)$ such that 
%	\begin{enumerate}
%		\item[i)] $s_i\leq 1$ for all $i\in N$
%		\item[ii)]  $A$ is EF1
%		\item[iii)] $||A_i|-|A_j||\leq 1$ for each $i,j\in N$
%	\end{enumerate}	
%\end{theorem}
%
%{\bf [...and where does this go?...]}

Finally, we show how to upper bound the transfer payments in the case of general valuation functions. Here, the welfare target is limited to the constant factor $\frac{1}{3}$, and the gap between our lower and upper bounds widens by a factor of $n$, but once again, this gap is independent of $m$.
\begin{theorem}\label{thm: upGen}
	For general valuations, for any $\alpha\in \left(0,\frac{1}{3}\right]$,  there is an envy-free allocation with transfers $(A,t)$ such 
	that $$\frac{\sw(A,t)}{\sw(A^*)}\geq \alpha \ \ \ $$ with total transfer $\sum_{i\in N} |t_i|\leq 2n^2\left(3\alpha \max_i v_i\left(A^*_i\right)+2\right)$.
\end{theorem}
\begin{proof}
We prove this result using an algorithm (see Algorithm~\ref{alg: addUpperGen}) that outputs the desired allocation with transfers $(A,t)$.
\begin{algorithm}[ht]
	\SetAlgoLined
	$B_i\gets \emptyset$ for all $i\in N$\;
	$A^*=(A^*_1,\cdots, A^*_n)\gets $ Welfare-Maximizing Allocation\;
	Let $\pi$ be an ordering of the agents with $A^*_{\pi(1)}\geq A^*_{\pi(2)} \geq \cdots \geq A^*_{\pi(n)} $\;
	\For{$k=1$ to $n$}{
		Let $\mathcal{S}$ be the set of all $X\subseteq M$ such that:\\
		$\quad\bullet$ $v_j(X)\geq 3\alpha v_{\pi(k)} (A^*_{\pi(k)})$ for some $j$ with $B_j=\emptyset$,\\
		$\quad\bullet$ $X\subseteq A^*_i$ for some $i$, and\\
		$\quad\bullet$ $X\cap B_\ell=\emptyset$ for all $\ell$\\
		\If{$|\mathcal{S}|\neq\emptyset$}
		{$\hat X_k \gets \argmin_{X\in \mathcal{S}} |X|$\;
		$B_j\gets \hat X_k$;}
	}
%	\While{$\exists g\in M\setminus \bigcup_{i\in N}B_i $}{Apply the cycle-elimination procedure of \citet{LMM04} to $G_B$ so that the resulting envy graphRemove all positive-weight cycles in $G_B$ using the envy-cycles procedure of \citet{LMM04}\;
%		Let $i$ be node in $G_B$ with $w_B(j,i)\leq 0$ for all $j\in N\setminus i$ \;
%		$B_i\gets B_i \cup g $
	Apply the envy-cycles procedure of \citet{LMM04} to allocate the items in $M\setminus \bigcup_{i\in N}B_i$\;
	Reassign bundles $B=(B_1,\cdots,B_n)$ to the agents to maximize the sum of utilities. Call this allocation $A$\;
	Compute the natural transfers $(t_1,\cdots,t_n)$\;
	\Return{$(A,t)$}
	\caption{Envy-free allocation with high welfare and small transfers for general valuations}\label{alg: addUpperGen}
\end{algorithm}

	We first show the bound on the transfer payments. Let $X=(X_1,\cdots,X_n)$ be the partial allocation obtained when the $\mathtt{for}$ loop finishes in Algorithm~\ref{alg: addUpperGen}. 
	Note that by the ordering of the optimal allocation, and by minimality of the allocated sets, we have, for any pair $i,j$ of agents, $v_i(X_j)\leq 3\alpha \max_i v_i(A^*_i)+1$. Thus $v_i(X_j)-v_i(X_i)\leq 3\alpha \max_i v_i\left(A^*_i\right)+1$. 
	At this stage, applying the envy-cycles procedure of \cite{LMM04} does not increase the envy by more than one. Let $B=(B_1,\cdots,B_n)$ be the partial allocation obtained after this step. We therefore have $v_i(B_j)-v_i(B_i)\leq 3\alpha \max_i v_i(A^*_i)+2$. Now, by Lemma~\ref{lem: EnvyMatch}, we have that $(A,t)$ is envy-free and $\sum_{i\in N} |t_i|\leq 2n^2\left(3\alpha \max_i v_i\left(A^*_i\right)+2\right)$.
	
	In order to show that $\frac{\operatorname{SW}(A,t)}{\operatorname{SW}(A^*)}\geq \alpha$, it suffices to show $\frac{\operatorname{SW}(X)}{\operatorname{SW}(A^*)}\geq \alpha$: since we add items to $X$ to obtain $A$, we have $\sw(X)\leq \sw(A)$, and since introducing transfers does not affect utilitarian welfare, we have $\sw(A,t)=\sw(A)$. Let $S\subseteq N$ be the set of time steps in which a set was allocated 
	during the $\mathtt{for}$ loop .
	The welfare $\sw(X)$ then satisfies 
	\begin{align}\label{eq:S}
		\sw(X) \ \ge\  \sum_{k\in S} 3\alpha\cdot v_{\pi(k)}(A^*_{\pi(k)})
	\end{align}
	Next consider rounds $k\in N\setminus S$, that is, the rounds when a bundle is not allocated. Since agent $\pi(k)$ can otherwise be allocated the set $A^*_{\pi(k)}$, if no set is allocated in round $k$ then either agent $\pi(k)$ has already received a set $X_{\pi(k)}$ of value at least $ 3\alpha \cdot v_{\pi(k)}(A^*_{\pi(k)}) $ 
	or some other agent who came before her 
	 received a set $X_{f(k)}\subseteq A^*_{\pi(k)}$ of value at least $ 3\alpha \cdot v_{\pi(k)}(A^*_{\pi(k)}) $.
	Thus $\max \left[ v_{f(k)}(X_{f(k)}), v_{\pi(k)}(X_{\pi(k)}) \right] \ge \alpha \cdot v_{\pi(k)}(A^*_{\pi(k)})$ and so
	\begin{align}\label{eq:S-bar}
		\sum_{k\in N\setminus S} \alpha \cdot v_{\pi(k)}(A^*_{\pi(k)}) 
		&\ \le\ \sum_{k\in N\setminus S}  \max \left[ v_{f(k)}(X_{f(k)}), v_{\pi(k)}(X_{\pi(k)}) \right] \nonumber \\
		&\ \le\  \sum_{k\in N\setminus S}  v_{f(k)}(X_{f(k)}) + v_{\pi(k)}(X_{\pi(k)}) \nonumber \\
		&\ \le\  2\cdot \sw(X)
	\end{align}
	Summing (\ref{eq:S}) and (\ref{eq:S-bar}) immediately gives the utilitarian welfare guarantee.
	\begin{align*}\label{eq:3-bound}
		3\sw(X)
		&\ \ge\ \sum_{k\in S} 3\alpha\cdot v_{\pi(k)}(A^*_{\pi(k)})+\sum_{k\in N\setminus S} 3\alpha \cdot v_{\pi(k)}(A^*_{\pi(k)}) \\
		&\ =\ 3\alpha\sum_{k\in N} v_{\pi(k)}(A^*_{\pi(k)}) \\
		&\ =\ 3\alpha\cdot \sw(A^*) \qedhere
	\end{align*}
%	The utilitarian welfare guarantee follows immediately.
\end{proof}
%Note that transfer payments of Algorithm~\ref{alg: addUpper} and Algorithm~\ref{alg: addUpperGen} are obtained by first 
%computing subsidy payments by maximum weight paths as seen in section~\ref{sec: smallSub}, and then subtracting the average subsidy. 